\documentclass[12pt, centerh1]{article}

\usepackage{amsfonts, amsmath, marvosym,colonequals,amssymb}
\usepackage{algorithm}
\usepackage{bm,relsize}
\usepackage{geometry}
\usepackage{lscape}
\usepackage{blkarray}
\usepackage{multirow,colonequals}
\usepackage{mathtools}
\usepackage{caption}
\usepackage{hyperref}
\usepackage{enumerate}
\usepackage{float}
\usepackage[section]{placeins}
\usepackage{subcaption}
\usepackage{titlesec} 
\usepackage{amsthm}
\usepackage[T1]{fontenc}
\usepackage{natbib}
\setcitestyle{sort&compress}
\newcommand{\GG}[1]{}

\newcommand{\vecx}{\mathbf{x}}

\newcommand{\vecX}{\mathbf{X}}

\newcommand{\vecQ}{\mathbf{Q}}
\newcommand{\vecq}{\mathbf{q}}
\newcommand{\vecA}{\mathbf{A}}
\newcommand{\vecV}{\mathbf{V}}

\newcommand{\vecv}{\mathbf{v}}

\newcommand\BibTeX{{\rmfamily B\kern-.05em \textsc{i\kern-.025em b}\kern-.08em
T\kern-.1667em\lower.7ex\hbox{E}\kern-.125emX}}
\newcommand{\bs}{\boldsymbol}

\newcommand{\E}{\mathbb{E}}
\DeclareMathOperator{\Tr}{tr}
\newcommand{\mvn}{\text{MVN}}
\newcommand{\gig}{\text{GIG}}

\newtheorem{Prop}{Proposition}

\newtheorem{Definition}{Definition}

  \title{Flexible Clustering for High-Dimensional Data via Mixtures of Joint Generalized Hyperbolic Models}
  
  \author{Yang Tang$^*$, Ryan P.\ Browne$^{**}$ and Paul D.\ McNicholas$^*$}
\date{\small $^{*}$Dept.\ of Mathematics \& Statistics, McMaster University, Hamilton, Ontario, Canada.\\
$^{**}$Dept.\ of Statistics and Actuarial Sciences, University of Waterloo, Ontario, Canada.}

\begin{document}
\maketitle
\begin{abstract}
\noindent A mixture of joint generalized hyperbolic distributions (MJGHD) is introduced for asymmetric clustering for high-dimensional data. The MJGHD approach takes into account the cluster-specific subspace, thereby limiting the number of parameters to estimate while also facilitating visualization of results. Identifiability is discussed, and a multi-cycle ECM algorithm is outlined for parameter estimation. The MJGHD approach is illustrated on two real data sets, where the Bayesian information criterion is used for model selection.\\[-6pt]

\end{abstract}

\section{Introduction}\label{sec:Intro}
Broadly, cluster analysis is the organization of a data set into meaningful clusters (or groups). Clustering in high-dimensional spaces has received increasing attention over the past few years because data collection has become easier and faster due to technological advances. Traditional clustering algorithms take all of the dimensions of a data set into account.
However, with high-dimensional data, the presence of irrelevant and noisy features can give misleading clustering results. In addition, data may be sparse as the number of dimensions increases, which is known as the ``curse of dimensionality'' \citep{bellman57}. 
Model-based clustering is a principled statistical approach for clustering, where data are clustered using some assumed mixture modelling structure \cite[see][for recent reviews and details]{bouveyron14,mcnicholas16a,mcnicholas16b}. A finite mixture model is a convex linear combination of a finite number of component distributions. Popular clustering methods are based on the Gaussian mixture model \citep[e.g.,][]{celeux95}, which assume that each class is represented by a Gaussian probability density. A parametrization of the component covariance matrices $\bs \Sigma_1, \ldots, \bs \Sigma_G$ via eigen-decomposition has been considered \citep{banfield93, celeux95}. The parametrization of the $p \times p$ component covariance matrices via eigen-decomposition is $\bs \Sigma_g=\lambda_g \vecQ_g \vecA_g \vecQ'_g$, where $\lambda_g=|\bs \Sigma_g|^{\frac{1}{p}}$, $\vecQ_g$ is the matrix of eigenvectors of $\bs \Sigma_g$, and $\vecA_g$ is a diagonal matrix, such that $|\vecA_g|=1$, with the normalized eigenvalues of $\bs \Sigma_g$ on the diagonal in decreasing order. 
The Gaussian parsimonious clustering models (GPCM) family \citep{celeux95} contains 14 parameterizations of $\bs \Sigma_g$ that result from imposing various constraints on $\lambda_g$, $\vecQ_g$, and $\vecA_g$. 
However, the parameterization used in the GPCMs cannot solve the problem of the \emph{curse of dimensionality} \cite[see][Sec.~2.3]{mcnicholas16a}. Feature transformation is another popular method for dimension reduction which builds new variables carrying a large part of the global information. For example, \citet{tipping99a} introduce probabilistic principal component analysis (PPCA) to find the principal subspace of the data and \cite{tipping99b} use a mixture of PPCA for clustering. Mixtures of factor analyzers \citep{ghahramani97,mclachlan00} and extensions thereof \citep[e.g.,][]{mcnicholas08} assume a lower dimension latent factor space. \citet{bouveyron07} propose a high-dimensional data clustering (HDDC) approach that encompasses both approaches.

Recent model-based clustering work has focused on mixtures of non-elliptical distributions \citep[e.g.,][]{browne15,lin16}. Dimensionality reduction approaches based on non-elliptical distributions have received relatively little attention, and recent work includes \citet{morris13,morris16}, \citet{murray14a,murray14b}, \cite{tortora16}, and \citet{lin16}. Each of these methods works well with particular types of data sets. However, the generalized hyperbolic distribution (GHD) represents perhaps the most flexible among the recent series of alternatives to the Gaussian component density \citep[see][]{browne15}.
We propose a joint generalized hyperbolic distribution (JGHD), which exhibits different marginal amounts of tail-weight. Moreover, it takes into account the component-specific subspace and, therefore, limits the number of parameters to estimate. This is a novel approach, which is applicable to high, and potentially very-high, dimensional spaces and with arbitrary correlation between dimensions. A multi-cycle expectation-conditional maximization (ECM) algorithm \citep{meng93} is used for parameter estimation and the Bayesian information criterion \citep[BIC;][]{schwarz78} is used to determine the number of components and the dimensions of the subspaces. This method is a robust asymmetric clustering method for high-dimensional data --- ``asymmetric'' in the sense that the clusters can be asymmetric. Our proposed method is illustrated, and compared to some comparitor clustering methods, on two real data sets.

\section{A Mixture of JGHDs for High-Dimensional Clustering}\label{sec:Method}
\subsection{Model-Based Subspace Clustering via Gaussian Mixtures} \label{sec:Overview}
A unified approach for model-based subspace clustering is introduced by \citet{bouveyron07}. Within the Gaussian mixture model framework, this approach assumes that class conditional densities are Gaussian $\text{MVN}(\bs \mu_g, \bs \Sigma_g)$ for $g=1,\ldots,G$. Let $\bs\Gamma_g$ consist of the eigenvectors of $\bs \Sigma_g$ as columns and $\bs \Phi_g$ is the diagonal matrix of the eigenvalues. Then the component covariance matrices $\bs \Sigma_g$ can be written $\bs \Sigma_g=\bs \Gamma_g\bs \Phi_g \bs\Gamma_g'$, for $g=1,\ldots,G$, where $\bs \Phi_g$ is divided into two blocks:
\begin{equation*}
\bs \Phi_g= \begin{blockarray}{(ccc|ccc)l}
	\phi_{g1}&&0&\BAmulticolumn{3}{c}{\multirow{3}{*}{\huge$0$}}&\\
	&\ddots&&&&&\\0&& \phi_{gq_g}&&&&\\
        \cline{1-6}
        \BAmulticolumn{3}{c|}{\multirow{3}{*}{\huge$0$}}&b_g&&0&\\
        &&&&\ddots&&\\&&&0&&b_g&\\
    \end{blockarray}
\end{equation*}
with $\phi_{gj}>b_g$, $j= 1,\ldots, q_g$, and $q_g \ll p$. The component-specific subspace $\mathcal {E}_g$ is defined as the affine space rotated by the $q_g$ eigenvectors associated with the eigenvalues $\bs \phi_{g}$. An EM algorithm is used for parameter estimation \citep[see][]{bouveyron07}. 

\subsection{A Multiple Scaled Generalized Hyperbolic Distribution} \label{sec:Overview}
A mixture of multiple scaled t-distributions is developed by \citet{forbes14}. The key elements of their approach are the introduction of multiple weight parameters and a decomposition of the matrix $\bs \Sigma=\bs \Gamma \bs \Phi \bs \Gamma'$, where $\bs \Gamma$ is the matrix of eigenvectors of $\bs \Sigma$ and $\bs \Phi$ is a diagonal matrix with the corresponding eigenvalues of $\bs \Sigma$. \citet[][]{browne15} use a mixture of GHDs; the GHD is a flexible distribution, capable of handling skewness and heavy tails, and has many well known distributions as special or limiting cases. \citet{tortora14a} introduce a mixture of multiple scaled generalized hyperbolic distributions, which is a more flexible model that forms the basis of our approach. A $p$-dimensional random variable $\vecX$ from a multiple-scaled GHD can be generated via, $\vecX=\bs \Gamma \bs \mu+ \bs \Gamma \bs \Delta_{w} \bs \beta+\bs \Gamma \vecV$,
where $\vecV\sim \mvn(\bs 0,\bs \Delta_w \bs \Phi)$ and $\bs \Delta_w=\text{diag}( w_1, \ldots, w_p)$. Therefore, $\vecX|\bs \Delta_w \sim \mvn(\bs \Gamma \bs \mu+\bs \Gamma\bs \Delta_w \bs \beta, \bs \Gamma\bs \Delta_w\bs \Phi \bs \Gamma')$ and the density of $\vecX$ can be written
\begin{equation*}
\begin{split}
f_{\text{\tiny MSGH}} (\vecx|\bs \mu, &\bs \Gamma, \bs \Phi,\bs \beta, \bs \Omega, \bs \lambda) =\\
&\int_{0}^{\infty}\cdots\int_{0}^{\infty}f_p\left(\bs \Gamma' \vecX-\bs \mu-\bs \Delta_w\bs \beta|\bs 0, \bs \Delta_w \bs \Phi\right)h_{\bs W} (w_1, \ldots, w_p| \bs \Omega, \bs 1, \bs \lambda)d \bs w,
\end{split}
\end{equation*}
where $f_p\left(\bs \Gamma' \vecX-\bs \mu-\bs \Delta_w\bs \beta|\bs 0, \bs \Delta_w \bs \Phi\right)$ is the density of a $p$-variate Gaussian distribution with mean $\bs 0$ and covariance $\bs \Delta_w \bs \Phi$, and $h_{\bs W}(w_1, \ldots, w_p| \bs \Omega, \bs 1, \bs \lambda)$ is a product of unidimensional generalized inverse Gaussian (GIG) distributions:
\begin{equation*}
h_{\bs W}(w_1, \ldots, w_p| \bs \Omega, \bs 1, \bs \lambda)=\prod_{j=1}^{p}\left[\frac{w_j^{\lambda_j-1}}{2K_{\lambda_j}(\Omega_j)}\exp\left\{-\frac{\Omega_j}{2}\left(w_j+\frac{1}{w_j}\right)\right\}\right].
\end{equation*}
Note that a related approach is discussed by \cite{wraith15}.





\subsection{A Mixture of JGHDs}\label{subspace}
Applying mixtures with flexible component densities, e.g., mixture of GHDs, to high-dimensional data is an important problem. Drawing ideas from model-based subspace clustering \citep{bouveyron07}, the JGHD chooses to project $p$-dimensional $\vecX$ onto two subspaces. We assume there is a $q$-dimensional subspace that best preserves the variance of the data and is much smaller than the original space. 
A $q$-dimensional weight variable $\mathbf W$ is incorporated into the density function of the first $q$ dimensions of $[\bs \Gamma'\vecX]$, where $\bs \Gamma$ is a matrix of eigenvectors associated with the eigenvalues $\bs \Phi=(\phi_1, \phi_2, \ldots, \phi_p)$, with $\phi_1>\phi_2>\cdots>\phi_p$, and $\bs \lambda=(\lambda_1, \ldots, \lambda_q)'$ is a $q$-dimensional index parameter. In addition, outside the $q$-dimensional subspace, the noise variance is modelled by a single parameter $b$ and a univariate latent variable $A$, where $A \sim \text{GIG}(\omega_0, 1, \lambda_0)$. Therefore, the JGHD takes the form
{\small\begin{equation*}
\begin{split}
&f(\vecx|\bs \mu, \bs \beta, \bs \Gamma, \bs \phi, b, \bs \Omega, \bs \lambda, \omega_0, \lambda_0)\\&=\prod_{j=1}^q\int_0^{\infty}\rho_1([\bs \Gamma'\vecx-\bs \mu-\bs \Delta_w \bs \beta]_j|\bs 0, \phi_j w_j)h_W(w_j|\Omega_j, 1, \lambda_j) d w_j\\&
\qquad\qquad\qquad\qquad\times  \int_0^{\infty}\prod_{k=q+1}^p\rho_1([\bs \Gamma' \vecx-\bs \mu- a \bs \beta]_k|\bs 0, b  a)h_A(a|\omega_0, 1, \lambda_0)d a\\
&=\prod_{j=1}^q\left[\frac{\Omega_j+\phi_j^{-1}([\bs \Gamma' \vecx]_j-\mu_j)^2}{\Omega_j+\beta_j^2\phi_j^{-1}}\right]^{\frac{\lambda_j-\frac{1}{2}}{2}}\frac{K_{\lambda_{j}-\frac{1}{2}}\left(\sqrt{\left[\Omega_j+\beta_j^2\phi_j^{-1}\right]\left[\Omega_j+\phi_j^{-1}\left([\bs \Gamma'\vecx]_j-\mu_j\right)^2\right]}\right)}{(2\pi)^{\frac{1}{2}}\phi_j^{\frac{1}{2}}K_{\lambda_j}(\Omega_j)\exp\{-(([\bs \Gamma' \vecx]_j-\mu_j)\beta_j)/\phi_j\}}\\ &\qquad\qquad\qquad\times\left[\frac{\omega_0+b^{-1}\sum_{k=q+1}^p([\bs \Gamma'\vecx]_k-\mu_k)^2}{\omega_0+b^{-1}\sum_{k=q+1}^p\beta_k^2}\right]^{\frac{(\lambda_0-\frac{p-q}{2})}{2}}\\
&\qquad\qquad\qquad\times\frac{K_{\lambda_0-\frac{p-q}{2}}\left(\sqrt{\left[\omega_0+b^{-1}\sum_{k=q+1}^p\beta_k^2\right]\left[\omega_0+b^{-1}\sum_{k=q+1}^p([\bs \Gamma'\vecx]_k-\mu_k)^2\right]}\right)}{(2\pi)^{\frac{p-q}{2}}b^{\frac{p-q}{2}}K_{\lambda_0}(\omega_0)\exp\left\{-(1/b){\sum_{k=q+1}^{p}([\bs \Gamma' \vecx]_k-\mu_k)\beta_k}\right\}}.
\end{split}
\end{equation*}}
Therefore, $W_{j}|\vecx \sim \gig (\Omega_j+\beta_j^2\phi_j^{-1}, \Omega_j+\{[\bs \Gamma' \vecx]_j-\mu_j\}/\phi_j,\lambda_j-{1}/{2})$
and $$A|\vecx \sim \gig\left(\omega_0+b^{-1}\sum_{k=q+1}^p\beta_k^2,  \omega_0+b^{-1}\sum_{k=q+1}^p([\bs \Gamma'\vecx]_k-\mu_k)^2, \lambda_0-\frac{p-q}{2}\right).$$

We use a mixture of JGHDs (MJGHD) for model-based clustering and classification. The density of the MJGHD is given by 
$f(\vecx|\bs \Psi)=\sum_{g=1}^G\pi_g f_{\text{JGHD}}(\vecx|\bs \Gamma_g, \bs \mu_g, \bs \beta_g, \bs \phi_g, b, \bs \Omega_g, \bs \lambda_g,\omega_{0g}, \lambda_{0g} )$,
in which we assume component-specific subspaces and the dimension $q_g$ of the subspace for the $g$th component can be considered to be the number of dimensions required to describe the main features of the $g$th component. 

\subsection{Parameter Estimation} \label{sec:par_est}	
To fit the models, we adopt the multi-cycle expectation-conditional maximization (ECM) algorithm of \citet{meng93}, which is a variant of the well-known expectation-maximization (EM) algorithm \citep{dempster77}.  The ECM algorithm exploits the simpler complete-data conditional maximization and replaces a complicated M-step with several CM-steps. In our case, the missing data comprise the group memberships $z_{ig}$, where $z_{ig}=1$ if observation $i$ belongs to component $g$ and $z_{ig}=0$ otherwise. The multidimensional latent variables $\bs \Delta_{W_{g}}=\text{diag}(W_{1g}, \ldots, W_{q_gg}, A_{g}\mathbf I_{p-q_g})$, $g=1,\ldots, G$, are assumed to follow GIG distributions. Therefore, the complete-data consist of the observed $\vecx_i$ together with the $z_{ig}$ and the $\bs \Delta_{W_{ig}}$, and the complete-data log-likelihood is given by:
$l_c(\bs \Psi)=l_{1c}(\bs \pi)+l_{2c}(\bs \theta)+l_{3c}(\bs \upsilon)+l_{4c}(\bs \tau)$,
where $l_{1c}(\bs \pi)=\sum_{i=1}^n\sum_{g=1}^Gz_{ig}\log \pi_g$, $l_{2c}(\bs \theta)=\sum_{i=1}^n\sum_{g=1}^G z_{ig}\log f_{p}\left( [\bs \Gamma_{g}' \vecx_i]|\bs \mu_{g}+\bs \Delta_{w_{ig}}\bs\beta_{g}, \bs \Delta_{w_{ig}}\bs \Phi_g\right)$, $l_{3c}(\bs \upsilon)=\sum_{i=1}^n\sum_{g=1}^G z_{ig}\sum_{j=1}^{q_g}$ $\log h_W(w_{ijg}|\Omega_{jg},1, \lambda_{jg})$, $l_{4c}(\bs \tau)=\sum_{i=1}^n\sum_{g=1}^G z_{ig}\log h_A(a_{ig}|\omega_{0g},1, \lambda_{0g})$, where $\bs \pi=(\pi_1, \ldots, \pi_G)$, $f_{p}\left( [\bs \Gamma_{g}' \vecx_i]|\bs \mu_{g}+\bs \Delta_{w_{ig}}\bs\beta_{g}, \bs \Delta_{w_{ig}}\bs \Phi_g\right)$ is the density of a multivariate Gaussian distribution with mean $\bs \mu_{g}+\bs \Delta_{w_{ig}}\bs\beta_{g}$ and covariance matrix $\bs \Phi_g=\text{diag}(\phi_1, \phi_2, \ldots, \phi_{q_g}, b_g\bs I_{p-q_{g}})$;
accordingly, $\bs \theta=\{\bs \Gamma_g, \bs \mu_g, \bs \beta_g, \bs \phi_g, b_g\}_{g=1}^G$.
Also, $\bs \upsilon=\{\bs\Omega_g, \bs \lambda_g\}_{g=1}^G$ and $\bs \tau=\{\omega_{0g}, \lambda_{0g}\}_{g=1}^G$.
The multi-cycle ECM algorithm used herein has two CM-steps on each iteration and an E-step is performed before each CM-step. 
They arise from the partition $\bs\Psi=(\bs\Psi_1, \bs\Psi_2)$, where $\bs\Psi_1=(\pi_g,\bs \mu_g, \bs \beta_g, \bs \phi_g, b_g, \bs \Omega_g, \bs \lambda_g, \omega_{0g}, \lambda_{0g})$ and $\bs \Psi_2=\bs\Gamma_g$.

\textbf{The E-step.}
We compute the expected value of the complete-data log-likelihood in the E-step using the expected values of the missing data in $l_c(\bs \Psi)$. We require the following expectations:\begin{equation*}
\begin{split}
&\E\left[Z_{ig}|\vecx_i\right]=\frac{\pi_g f(\vecx_i|\bs \Psi_g)}{\sum_{h=1}^G\pi_h f(\vecx_i|\bs \Psi_h)}\equalscolon\hat{z}_{ig}, \\
&\E\left[W_{ijg}|\vecx_i, z_{ig}=1\right]=\sqrt{\frac{e_{ijg}}{d_{jg}}}\frac{K_{\lambda_{jg}+1/2}(\sqrt{e_{ijg}d_{jg}})}{K_{\lambda_{jg}-1/2}(\sqrt{e_{ijg}d_{jg}})}\equalscolon E_{1ijg},\\
&\E\left[W^2_{ijg}|\vecx_i, z_{ig}=1\right]=\frac{e_{ijg}}{d_{jg}}\frac{K_{\lambda_{jg}+3/2}(\sqrt{e_{ijg}d_{jg}})}{K_{\lambda_{jg}-1/2}(\sqrt{e_{ijg}d_{jg}})}\equalscolon E_{2ijg},\\
&\E\left[{1}/{W_{ijg}}|\vecx_i, z_{ig}=1\right]=\sqrt{\frac{d_{jg}}{e_{ijg}}}\frac{K_{\lambda_{jg}+1/2}(\sqrt{e_{ijg}d_{jg}})}{K_{\lambda_{jg}-1/2}(\sqrt{e_{ijg}d_{jg}})}-\frac{2\lambda_{jg}-1}{e_{ijg}}\equalscolon E_{3ijg},\\
&\E\left[\log W_{ijg}|\vecx_i, z_{ig}=1\right]=\log\sqrt{\frac{e_{ijg}}{d_{jg}}}+\frac{\partial}{\partial \upsilon}\log \left(K_\upsilon (\sqrt{e_{ijg}d_{jg}})\right)|_{\upsilon=\lambda_{jg}-1/2}\equalscolon E_{4ijg},\\
\end{split}
\end{equation*}
where $d_{jg}=\Omega_{jg}+\beta_{jg}^2\phi_{jg}^{-1}$ and $e_{ijg}=\Omega_{jg}+\{[\bs \Gamma_{g}' \vecx_i]_j-\mu_{jg}\}/{\phi_{jg}}$. We also require:
\begin{equation*}
\begin{split}
&\E\left[A_{ig}|\vecx_i, z_{ig}=1\right]=\sqrt{\frac{e_{0ig}}{d_{0g}}}\frac{K_{\lambda_{0g}-(p-q_g)/2+1}(\sqrt{e_{0ig}d_{0g}})}{K_{\lambda_{0g}-(p-q_g)/2}(\sqrt{e_{0ig}d_{0g}})}\equalscolon J_{1ig},\\
&\E\left[A^2_{ig}|\vecx_i, z_{ig}=1\right]=\frac{e_{0ig}}{d_{0g}}\frac{K_{\lambda_{0g}-(p-q_g)/2+2}(\sqrt{e_{0ig}d_{0g}})}{K_{\lambda_{0g}-(p-q_g)/2}(\sqrt{e_{0ig}d_{0g}})}\equalscolon J_{2ig},\\
&\E\left[{1}/{A_{ig}}|\vecx_i, z_{ig}=1\right]=\sqrt{\frac{d_{0g}}{e_{0ig}}}\frac{K_{\lambda_{0g}-(p-q_g)/2+1}(\sqrt{e_{0ig}d_{0g}})}{K_{\lambda_{0g}-(p-q_g)/2}(\sqrt{e_{0ig}d_{0g}})}-\frac{2\lambda_{0g}-(p-q_g)}{e_{0ig}}\equalscolon J_{3ig},\\
&\E\left[\log A_{ig}|\vecx_i, z_{ig}=1\right]=\log\sqrt{\frac{e_{0ig}}{d_{0g}}}+\frac{\partial}{\partial \upsilon}\log \left(K_\upsilon (\sqrt{e_{0ig}d_{0g}})\right)|_{\upsilon=\lambda_{0g}-(p-q_g)/2}\equalscolon J_{4ig},
\end{split}
\end{equation*}
where $d_{0g}= \omega_{0g}+b_{g}^{-1}\sum_{k=q_g+1}^p\beta_{kg}^2$, and $e_{0ig}=\Omega_{0g}+b_{g}^{-1}\sum_{k=q_g+1}^p([\bs \Gamma_{g}'\vecx_i]_k-\mu_{kg})^2$.
Thus we have 
$\E[\bs \Delta_{W_{ig}}]=\text{diag}(E_{1i1g}, E_{1i2g},\ldots, E_{1iq_gg},J_{1ig}\bs I_{p-q_g} )$,
$\E[\bs \Delta_{{1}/{W_{ig}}}]=\text{diag}(E_{3i1g},E_{3i2g},\ldots, E_{3iq_gg},J_{3ig} \bs I_{p-q_g})$,
$\E[\bs \Delta_{W^2_{ig}}]=\text{diag}(E_{2i1g}, E_{2i2g},\ldots,E_{2iq_gg},
 J_{2ig}\bs I_{p-q_g})$.

\textbf{CM-step 1.} 
The first CM-step on the $(t+1)$th iteration requires the calculation of $\bs \Psi_1^{(t+1)}$ as the value of $\bs \Psi_1$ that maximizes $Q(\bs\Psi|\bs\Psi^{(t)})$ with $\bs \Psi_2$ fixed at $\bs \Psi_2^{(t)}$. In particular, we obtain the update for the mixing proportions as $\hat{\pi}_g^{(t+1)}=n_g^{(t)}/n$, where $n_g=\sum_{i=1}^n \hat{z}^{(t)}_{ig}$.
The elements of the location parameter $\bs \mu_g$ and skewness parameter $\bs \beta_g$ are replaced with 
\begin{equation*}
\begin{split}
&\mu_{jg}^{(t+1)}=\frac{\sum_{i=1}^n \hat{z}_{ig}^{(t)}[\bs\Gamma_{g}'^{(t)}\vecx_i]_j \left(\frac{\sum_{i=1}^n \hat{z}_{ig}^{(t)}\E[\bs \Delta_{W_{ig}}]_j^{(t)}}{n_g^{(t)}} \E[\bs \Delta_{{1}/{W_{ig}}}]_j^{(t)}-1\right)  }{\sum_{i=1}^n\hat{z}_{ig}^{(t)}\left(\frac{\sum_{i=1}^n \hat{z}_{ig}^{(t)}\E[\bs \Delta_{W_{ig}}]_j^{(t)}}{n_g^{(t)}} \E[\bs \Delta_{{1}/{W_{ig}}}]_j^{(t)}-1\right)}
\end{split}
\end{equation*}
and
\begin{equation*}
\begin{split}
\beta_{jg}^{(t+1)}=\frac{\sum_{i=1}^n \hat{z}_{ig}^{(t)}[\bs \Gamma_{g}'^{(t)}\vecx_i]_j \left(\frac{\sum_{i=1}^n \hat{z}_{ig}^{(t)}\E[\bs \Delta_{{1}/{W_{ig}}}]_j^{(t)}}{n_g^{(t)}}- \E[\bs \Delta_{{1}/{W_{ig}}}]_j^{(t)}\right)  }{\sum_{i=1}^n\hat{z}_{ig}^{(t)}\left(\frac{\sum_{i=1}^n \hat{z}_{ig}^{(t)}\E[\bs \Delta_{W_{ig}}]_j^{(t)}}{n_g^{(t)}} \E[\bs \Delta_{{1}/{W_{ig}}}]_j^{(t)}-1\right)},
\end{split}
\end{equation*}
respectively, where $j=1,2, \ldots, p$ and $[\bs\Gamma_{g}'^{(t)}\vecx_i]_j$ is the $j${th} element of $[\bs\Gamma_{g}'^{(t)}\vecx_i]$. 
We update the diagonal elements $h_{jg}$ of the empirical covariance matrix of $[\bs\Gamma_{g}'\vecx]_j|\bs \Delta_{W_{g}}$ via
\begin{equation*}\begin{split}
h_{jg}^{(t+1)}=\frac{1}{n_g^{(t)}}\sum_{n=1}^n \Big\{\hat{z}_{ig}^{(t)}([\bs\Gamma_{g}'^{(t)}\vecx_i]_j&-\mu_{jg}^{(t+1)})^2-2 \hat{z}_{ig}^{(t)}([\bs\Gamma_{g}'^{(t)}\vecx_i]_j-\mu_{jg}^{(t+1)})\beta_{jg}^{(t+1)}\E[\bs \Delta_{W_{ig}}]_j^{(t)}\\&
\qquad\qquad\qquad\qquad\qquad\qquad+\hat z_{ig}^{(t)} (\E[\bs \Delta_{W^2_{ig}}]_j^{(t)}(\beta_{jg}^2)^{(t+1)}\Big\}.
\end{split}\end{equation*}
We then order $h_{jg}^{(t+1)}$ from the largest to the smallest in order to determine the subspaces. Now we obtain  
\begin{equation*}
\begin{split}
&\phi_{jg}^{(t+1)}=\frac{1}{n_g^{(t)}}\sum_{i=1}^n \hat{z}_{ig}^{(t)}\left[E_{3ijg}^{(t)}([\bs \Gamma_{g}'^{(t)}\vecx_i]_j-\mu^{(t+1)}_{jg})^2-2([\bs \Gamma_{g}'^{(t)}\vecx_i]_j-\mu_{jg}^{(t+1)})\beta_{jg}^{(t+1)}+E_{1ijg}^{(t)}(\beta^2_{jg})^{(t+1)}\right],\\
&b_g^{(t+1)}=\frac{1}{n_g^{(t)}(p-q_g)}\sum_{i=1}^n \hat {z}_{ig}^{(t)}\sum_{k=q_g+1}^p\left[J_{3ig}^{(t)}[\bs \Gamma_{g}'^{(t)}\vecx_i]^2_k+J_{3ig}^{(t)} (\mu_{kg}^2)^{(t+1)}+J_{1ig}^{(t)}(\beta_{kg}^2)^{(t+1)} \right.\\
& \qquad\qquad\qquad\qquad\qquad\qquad\qquad\qquad \left.-2J_{3ig}^{(t)}[\bs \Gamma_{g}'^{(t)}\vecx_i]_k\mu_{kg}^{(t+1)}-2[\bs \Gamma_{g}'^{(t)}\vecx_i]_k\beta_{kg}^{(t+1)}-2\mu_{kg}^{(t+1)}\beta_{kg}^{(t+1)}\right].
\end{split}
\end{equation*}
The $q_g$-dimensional concentration parameter  $\bs \Omega_g$ and index parameter $\bs \lambda_g$ are estimated by maximizing the function 
\begin{equation*}
q_{jg}(\Omega_{jg}, \lambda_{jg})=-\log K_{\lambda_{jg}}(\Omega_{jg})+(\lambda_{jg}-1)\frac{\sum_{i=1}^n \hat{z}_{ig}E_{4ijg}}{n_g}-\frac{\Omega_{jg}}{2}\left(\sum_{i=1}^n\hat{z}_{ig}E_{1ijg}+\sum_{i=1}^nz_{ig}E_{3ijg}\right).
\end{equation*}
This leads to
{\begin{equation*}
\lambda_{jg}^{(t+1)}=\frac{\sum_{i=1}^n \hat{z}_{ig}E_{4ijg}}{n_g}\lambda_{jg}^{(t)}\left[\frac{\partial}{\partial v}\log K_v(\Omega_{jg}^{(t)})|_{v=\lambda_{jg}^{(t)}}\right]^{-1}
\end{equation*}
and 
\begin{equation*}
\Omega_{jg}^{(t+1)}=\Omega_{jg}^{(t)}-\left[\frac{\partial}{\partial v}q_{jg}(v, \lambda_{jg}^{(t+1)})|_{v=\Omega_{jg}^{(t)}}\right]\left[\frac{\partial^2}{\partial v^2}q_{jg}(v, \lambda_{jg}^{(t+1)})|_{v=\Omega_{jg}^{(t)}}\right]^{-1}.
\end{equation*}}The univariate parameters $\omega_{0g}$ and $\lambda_{0g}$ are estimated as for the mixture of GHDs \citep[see][]{browne15}.

\textbf{CM-step 2.}
To update the component eigenvector matrices $\bs \Gamma_{g}$, our goal is to minimize the matrix trace function
\begin{equation*}\begin{split}
f(\bs \Gamma_{g})&=\frac{1}{2}\Tr\left\{\sum_{i=1}^n \hat{z}_{ig}\vecx_i \vecx_i'\bs \Gamma_{g}\hat{\bs\Phi}_g \E[\bs \Delta_{{1}/{W_{ig}}}] \bs \Gamma_{g}'\right\}\\&
\qquad\qquad\qquad\qquad\qquad-\Tr\left\{\sum_{i=1}^n\hat{z}_{ig}\bs \Phi_{g}^{-1}(\E[\bs \Delta_{{1}/{W_{ig}}}]\bs \mu_g+\bs \beta_g)\vecx_i'\bs \Gamma_{g}\right\}+\text{constant}.
\end{split}\end{equation*}
We follow \citet{kiers02} and \citet{browne14} by using a majorization function for the minimization of $f(\bs \Gamma_{g})$ and it takes the form
$f(\bs \Gamma_{g})\le\text{constant}+\Tr\left(\bs F_t \bs \Gamma_{g}\right)$,
where 
\begin{equation*}
\begin{split}
\bs F_t=&\sum_{i=1}^n\left(-\hat{z}_{ig}^{(t)}(\bs \Phi_g^{-1})^{(t+1)}(\E[\bs \Delta_{{1}/{W_{ig}}}]^{(t)}\bs \mu_g^{(t+1)}+\bs \beta_g^{(t+1)})\vecx_i'\right)\\&\qquad\qquad\qquad\qquad\qquad+\sum_{i=1}^n\left(\hat{z}_{ig}^{(t)}\vecx_i \vecx_i' \bs \Gamma_{g}\E[\bs \Delta_{{1}/{W_{ig}}}]^{(t)}(\bs\Phi_{g}^{-1})^{(t+1)}-\hat{z}_{ig}^{(t)}\alpha_{ig}^{(t+1)}\vecx_i \vecx_i' \bs \Gamma_{g}\right),
\end{split}
\end{equation*}
where $\bs \Phi_g^{(t+1)}=\text{diag}(\phi_{1g}^{(t+1)}, \phi_{2g}^{(t+1)},\ldots, \phi_{q_gg}^{(t+1)}, b_g^{(t+1)}\mathbf I_{p-d})$  and  $\alpha_{ig}$ is the largest value of the diagonal matrix $\E[\bs \Delta_{{1}/{W_{ig}}}]^{(t)}\bs (\bs \Phi_{g}^{-1})^{(t+1)}$. Suppose we obtain the singular value decomposition
$\bs -F_t=\mathbf P_t \mathbf B_t \mathbf R_t'$ in which $\mathbf P_t$ and $\mathbf R_t$ are orthonormal, and $\mathbf B_t$ is diagonal, then the update of $\bs \Gamma_{g}$ becomes $\bs \Gamma_{g}^{(t+1)}=\mathbf {R}_t\mathbf{P}_t'$.


%

\subsection {Model Identifiability}

The identifiability of our MJGHD is investigated in this section. 
The identifiability of the MJGHD depends on the identifiability of the mixture of univariate generalized hyperbolic distributions which has been proved in \citet{browne15}. In Proposition 1, we extend the results in \citet{browne15} and show that the MJGHD is identifiable assuming correct choice of $q_g$ $(g = 1,\ldots, G)$. 

\begin{Definition}\label{d1}
Let $\bs \Sigma$ be a square, symmetric real-valued $p\times p$ matrix with $p$ linearly independent eigenvectors. Then there exists a symmetric diagonal decomposition 
$\bs \Sigma=\bs\Gamma \bs \Phi \bs \Gamma'$,
where the columns of $\bs \Gamma$ are the orthogonal and normalized eigenvectors of $\bs \Sigma$, and $\bs \Phi$ is the diagonal matrix whose entries are the eigenvalues of $\bs \Sigma$. Further, all entries of $\bs\Gamma$ are real and we have $\bs\Gamma^{-1}=\bs\Gamma'$. 
\end{Definition}
\begin{Prop}\label{p1}
The JGHDs generate identifiable finite mixtures assuming the correct choice of $q_g$ $(g = 1,\ldots, G)$. 
\end{Prop}

\begin{proof}
Consider moving the amount $t$ in direction $\bs z$, setting $\vecX=t\bs z$. If $\bs z$ is equal to the $k${th} eigenvector ($k=1, \ldots q$) then the density reduces to 
\begin{equation*}
c_k\left[\frac{\Omega_k+\phi_k^{-1}([t-\mu_k)^2}{\Omega_k+\beta_k^2\phi_k^{-1}}\right]^{\frac{\lambda_k-\frac{1}{2}}{2}}\frac{K_{\lambda_{k}-\frac{1}{2}}\left(\sqrt{\left[\Omega_k+\beta_k^2\phi_k^{-1}\right]\left[\Omega_k+\phi_k^{-1}\left(t-\mu_k\right)^2\right]}\right)}{(2\pi)^{\frac{1}{2}}\phi_k^{\frac{1}{2}}K_{\lambda_k}(\Omega_k)\exp\{-{(t-\mu_k)\beta_k}/{\phi_k}\}}, 
\end{equation*}
where 
{\small\begin{equation*}
        \begin{split}
            &c_k=\prod_{j=1,j\neq k}^q\left[\frac{\Omega_j+\phi_j^{-1}\mu_j^2}{\Omega_j+\beta_j^2\phi_j^{-1}}\right]^{\frac{\lambda_j-\frac{1}{2}}{2}}\frac{K_{\lambda_{j}-\frac{1}{2}}\left(\sqrt{\left[\Omega_j+\beta_j^2\phi_j^{-1}\right]\left[\Omega_j+\phi_j^{-1}\mu_j^2\right]}\right)}{(2\pi)^{\frac{1}{2}}\phi_j^{\frac{1}{2}}K_{\lambda_j}(\Omega_j)\exp\{{\mu_j\beta_j}/{\phi_j}\}}\\& \ \times\left[\frac{\omega_0+b^{-1}\sum_{d=q+1}^p\mu_d^2}{\omega_0+b^{-1}\sum_{d=q+1}^p\beta_d^2}\right]^{\frac{(\lambda_0-\frac{p-q}{2})}{2}}\frac{K_{\lambda_0-\frac{p-q}{2}}\left(\sqrt{\left[\omega_0+b^{-1}\sum_{d=q+1}^p\beta_d^2\right]\left[\omega_0+b^{-1}\sum_{d=q+1}^p\mu_d^2\right]}\right)}{(2\pi)^{\frac{p-q}{2}}b^{\frac{p-q}{2}}K_{\lambda_0}(\omega_0)\exp\left\{\frac{1}{b}\sum_{d=q+1}^{p}\mu_d\beta_d\right\}}.
        \end{split}
    \end{equation*}}
    Therefore, the density
    {\begin{equation*}
        f(t|\bs \theta)\propto \left[\frac{\Omega_k+\phi_k^{-1}([t-\mu_k)^2}{\Omega_k+\beta_k^2\phi_k^{-1}}\right]^{\frac{\lambda_k-\frac{1}{2}}{2}}\frac{K_{\lambda_{k}-\frac{1}{2}}\left(\sqrt{\left[\Omega_k+\beta_k^2\phi_k^{-1}\right]\left[\Omega_k+\phi_k^{-1}\left(t-\mu_k\right)^2\right]}\right)}{(2\pi)^{\frac{1}{2}}\phi_k^{\frac{1}{2}}K_{\lambda_k}(\Omega_k)\exp\{-{(t-\mu_k)\beta_k}/{\phi_k}\}}.
    \end{equation*}}As \citet{browne15} note, if two (or more) parameterizations are one-to-one and one parameterization is identifiable, then the other is also identifiable. Set $\delta_j=\beta_j/\phi_j$, $\alpha_j=\sqrt{\Omega_j/\phi_j+\beta_j^2/\phi_j^2}$ and $\kappa_j=\sqrt{\phi_j\Omega_j}$. For large $z$, the Bessel function can approximated by 
\begin{equation*}
K_{\lambda(z)}=\sqrt{\frac{\pi}{2z}}e^{-z}\left[1+O\left(\frac{1}{z}\right)\right],
\end{equation*}
and the characteristic function for a normal variance-mean density can be written as 
$\varphi_X(t)=\exp\{it\mu\}M_W\left(\beta ti-\sigma^2t^2/2~|~\lambda, \Omega\right)$.
Therefore, the characteristic function for the JGHD can be written
{\small\begin{equation*}
\begin{split}
&\varphi_{\vecX}(\vecv)=\prod_{j=1}^q\exp\{i|\bs \Gamma' \vecv|_j \mu_j\}\left[1+\frac{\phi_j|\bs \Gamma' \vecv|_j^2-2\beta_j|\bs \Gamma' \vecv|_ji}{\Omega_j}\right]^{-\lambda_j/2}\frac{K_{\lambda_j}\left(\sqrt{\Omega_j\left[\Omega_j+(\phi_j|\bs \Gamma' \vecv|_j^2-2\beta_j|\bs \Gamma' \vecv|_ji)\right]}\right)}{K_{\lambda_j}(\Omega_j)}\\
&\quad \times\exp\{i|\bs \Gamma' \vecv|'_2\bs \mu_2\}\left[1+\frac{b|\bs \Gamma' \vecv|'_2|\bs \Gamma' \vecv|_2-2\bs \beta'_2|\bs \Gamma' \vecv|i}{\omega_0}\right]^{-\lambda_{0}/2}\frac{K_{\lambda_0}\left(\sqrt{\omega_0\left[\omega_0+(b|\bs \Gamma' \vecv|'_2|\bs \Gamma' \vecv|_2-2\bs \beta_2'|\bs \Gamma' \vecv|_2i)\right]}\right)}{K_{\lambda_0}(\omega_0)},
\end{split}
\end{equation*}}
where $|\bs \Gamma' \vecv|_2$ is the $(q+1)$th to $p$th columns of $|\bs \Gamma' \vecv|$, $\bs \mu_2=(\mu_{q+1}, \ldots, \mu_p)'$, and $\bs \beta_2=(\beta_{q+1}, \ldots, \beta_p)'$. Now if we consider moving $t$ in the direction $\bs z$, $\vecv=t\bs z$ and, for large $t$, the characteristic function is 
{\small\begin{equation*}
\begin{split}
\varphi_{\vecX}(t \bs z)&\propto\exp\Bigg\{it\sum_{j=1}^p|\bs \Gamma' \bs z|_j \mu_j-t\sum_{j=1}^q\kappa_j||\bs \Gamma' \bs z|_j|-t \sqrt{b\omega_0}\sum_{k=q+1}^p|\bs \Gamma' \bs z|_k\\&
\qquad\qquad\qquad\qquad\qquad\qquad\qquad\qquad\qquad\qquad-\log(t)\Bigg[\sum_{j=1}^q \lambda_j I (|\bs \Gamma' \bs z|_j \neq 0)+\lambda_0\Bigg]+O(1)\Bigg\}\\
&\propto\exp\Bigg\{it\bs z '\bs \Gamma \bs \mu  -t\left[\sum_{j=1}^q\kappa_j||\bs \Gamma' \bs z|_j|+\sqrt{b\omega_0}\sum_{k=q+1}^p|\bs \Gamma' \bs z|_k\right]\\&
\qquad\qquad\qquad\qquad\qquad\qquad\qquad\qquad\qquad\qquad -\log (t) \Bigg[\sum_{j=1}^q \lambda_j I (|\bs \Gamma' \bs z|_j \neq 0)+\lambda_0\Bigg]+O(1)\Bigg\}.
\end{split}
\end{equation*}}
From \cite{yakowitz68}, there exists $\bs z$ such that the tuple $$\left(\bs z '\bs \Gamma \bs \mu, \sum_{j=1}^q\kappa_j||\bs \Gamma' \bs z|_j|+\sqrt{b\omega_0}\sum_{k=q+1}^p|\bs \Gamma' \bs z|_k, \sum_{j=1}^q \lambda_j I (|\bs \Gamma' \bs z|_j \neq 0)+\lambda_0\right)$$ is pairwise distinct for all $g=1,\ldots, G$ and reduce to a mixture of univariate hyperbolic distributions, which is identifiable \citep{browne15}.
\end{proof}

\subsection{Computational Aspects}
We start with ten random initializations of the algorithm by randomly assigning each observation to one of the $G$ components. After fitting models for all values of $G$ and $q_g$, we use the BIC. We compare our approach with the classic Gaussian parsimonious clustering models (GPCM) from {\sf R} package \texttt{mixture} \citep{browne15b} and high-dimensional data clustering (HDDC) approach from {\sf R} package \texttt{HDclassif} \citep{berge12}. It is worth noting that the MJGHD proposed herein does not need to numerically invert covariance matrices, which often fails for singularity reasons. 
We compare with the parsimonious Gaussian mixture models from {\sf R} package \texttt{pgmm} \citep{mcnicholas11} and the mixture of generalized hyperbolic factor analyzers \citep[MGHFA;][]{tortora16} from {\sf R} package \texttt{mixGHD} \citep{tortora17} in our real data applications.

\section{Real Data}\label{sec:App.}
\subsection{Italian Wines}
The Italian wines data \citep{forina86} has been widely used in literature. The data set includes 27 physical and chemical properties of 178 wines and each wine belongs to one of the three types: Barolo, Grignolino or Barbera. This data set is available from the {\sf R} package \texttt{pgmm}. The MJGHD approach was fitted to these data for $G=1,2,\ldots, 4$ and $q_g= 2,3,5,8,10$. The highest BIC occurs at the three-component, $\vecq=(8,5,3)$ model. The BIC value is $-16984$. 
A summary of the best models from the MJGHD, HDDC, GPCM, PGMM, and the MGHFA approaches is shown in Table~\ref{table:4.2}. The MJGHD and the PGMM approaches yield excellent clustering results and outperform the chosen Gaussian mixture models, HDDC, and MGHFA approaches. The MJGHD misclassifies only three of the 178 wines (Table~\ref{table:4.3}). It is worth noting that MJGHD is one of the few methods in the literature that uses all 27 variables of the wine data set and yields excellent clustering results. 
\begin{table}[ht]
\centering
\vspace{-0.1cm}
\caption{A comparison of the selected MJGHD and four different approaches on the wine data.}
\vspace{-0.12in}
\scalebox{0.85}{
\begin{tabular*}{1\textwidth}{@{\extracolsep{\fill}}l r r r r}
\hline
&$G$&Model&BIC&ARI\\ [0.5ex] 
\hline
MJGHD&$3$&$\vecq=(8,5,3)$&$-16984$&$0.95$\\
HDDC&$2$&$\vecq=(2,1)$&$-12657$&$0.41$\\
PGMM&$3$&CUU, $q=4$&$-11428$&$0.96$\\
GPCM&$2$&EVE&$-12068$&$0.49$\\
MGHFA&$2$&$q=2$&$-12653$&$0.49$\\
\hline
\end{tabular*}}
\label{table:4.2}
\end{table}
\begin{table}[!htp]
\vspace{-0.3cm}	
\centering
\caption{Cross-tabulation of true versus predicted (A,B,C) classifications from the selected MJGHD for the wine data.}
\vspace{-0.12in}
\scalebox{0.85}{
\begin{tabular*}{1\textwidth}{@{\extracolsep{\fill}} lcccr}
\hline
&A&B&C&ARI\\
\hline
Barolo&59&0&0&\multirow{3}{*}{0.95}\\
Grignolino&2&68&1\\
Barbera&0&0&48\\
\hline
\end{tabular*}}
\label{table:4.3}
\end{table}

\subsection{Breast Cancer Diagnostic Data Set}
The breast cancer diagnostic data was originally reported on by \citet{street93}. They give data on 569 cases of breast tumours --- 357 benign and 212 malignant --- and ten real-valued features are computed for each cell nucleus. The mean, standard error, and the ``worst'' or the largest of these features were computed for each image, resulting in 30 attributes. For instance, Attribute~3 is mean radius, Attribute~13 is the standard error of radius and Attribute~23 is the worst radius. The MJGHD approach is fitted to these data for $G=1,2,\ldots, 4$ and $q_g=2,3,5,8,10$. The highest BIC occurs at the two-component, $\vecq=(8,5)$ model. The BIC value is $-20432$. 
%
A summary of the best models from the MJGHD, HDDC, PGMM, GPCM, and MGHFA approaches is shown in Table~\ref{table:4.5}. The respective classification results reveal that the chosen two-component MJGHD model yields relatively good clustering result (ARI$=0.70$) and outperforms the approaches we compared with. Moreover, the MJGHD approach is the only one that gives the correct number of components. Plots of the first three dimensions of the transformed spaces for each component with group labels (Figure~\ref{fig:4.1}) indicate that 
the components are well separated in the latent space.
\begin{table}[htb]
\centering
\vspace{-0.1cm}
\caption{A comparison of the selected MJGHD and four different approaches on the tumour data.}
\vspace{-0.12in}
\scalebox{0.85}{
\begin{tabular*}{1\textwidth}{@{\extracolsep{\fill}}lrrrr}
\hline
&$G$&Model&BIC&ARI\\ [0.5ex] 
\hline
MJGHD&$2$&$\vecq=(8,5)$&$-20432$&$0.70$\\
HDDC&$4$&$\vecq=(4,3,3,3)$&$-26673$&$0.09$\\
PGMM&$4$& UUU, $q=4$& $-12083$&$0.35$\\
GPCM&$4$&VEE&$-24367$&$0.22$\\
MGHFA&$5$&$q=4$&$-13777$&$0.58$\\
\hline
\end{tabular*}}
\label{table:4.5}
\end{table}
\begin{table}[htb]
\centering
\vspace{-0.1cm}
\caption{Cross-tabulation of true versus predicted (A,B) classifications from the selected MJGHD for the tumour data.}
\vspace{-0.12in}
\scalebox{0.85}{
\begin{tabular*}{1\textwidth}{@{\extracolsep{\fill}} lccr}
\hline
&A&B&ARI\\
\hline
Malignant&343&15&\multirow{2}{*}{0.70}\\
Benign&30&164&\\
\hline
\end{tabular*}}
\label{table:4.6}
\end{table}

\begin{figure}[!htb]
        \centering
        \vspace{-1cm}
        \begin{subfigure}[b]{0.45\textwidth}
                \includegraphics[width=\textwidth]{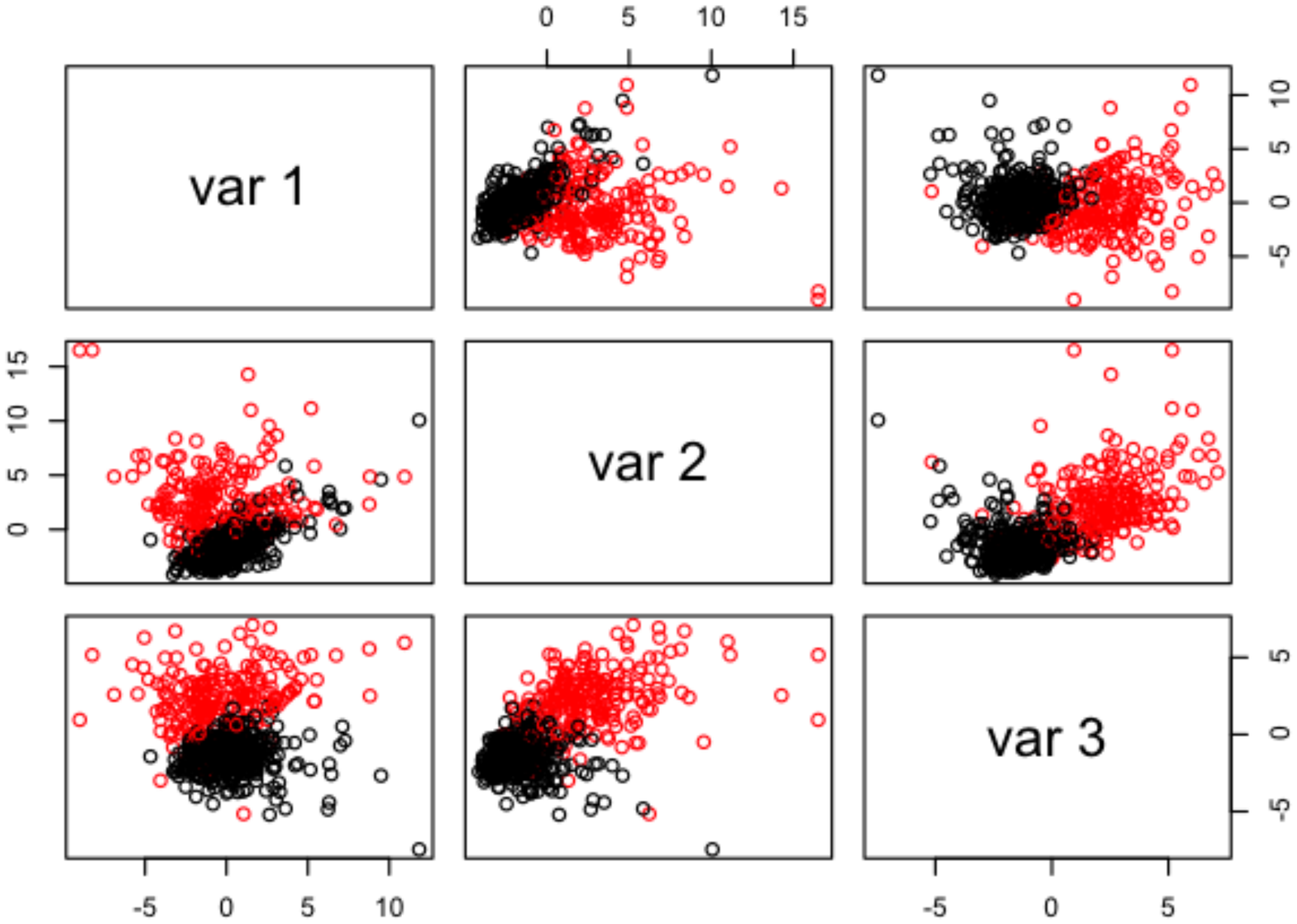}
               \end{subfigure}%
  ~      \begin{subfigure}[b]{0.45\textwidth}
                \includegraphics[width=\textwidth]{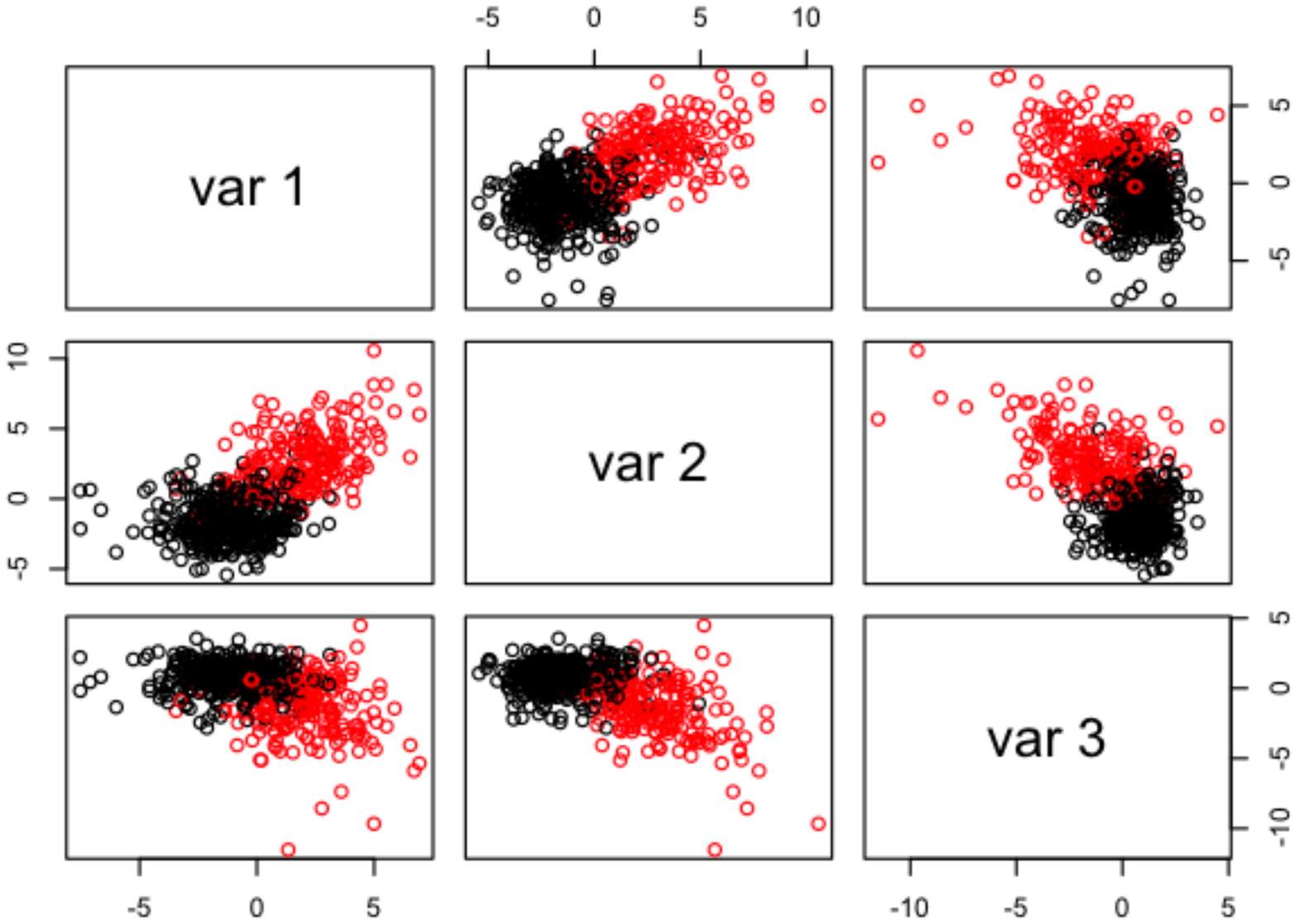}
               \end{subfigure}
                \vspace{-1cm}
                \caption{The first three dimensions of the transformed spaces for components 1 (left) and~2 (right) from the selected MJGHD model.}
\label{fig:4.1}
\end{figure}

\section{Conclusion}\label{sec:Con.}
The MJGHD approach for asymmetric clustering of high-dimensional data was introduced. We developed the MJGHD based on a mixture of GHDs, which represents perhaps the most flexible in a recent series of alternatives to the Gaussian mixture model for clustering and classification. For one, allowing the dimension of the component-specific subspace to vary across components provides quite some flexibility. Parameter estimation was carried out using a multi-cycle ECM algorithm and, notably, the MJGHD approach does not require numerical inversion of covariance matrices. The BIC was used for model selection. 
Comparing the MJGHD, HDDC, MGHFA, PGMM, and GPCM approaches yielded some interesting results. Two real data sets were considered for illustration: the Italian wine data and the breast cancer diagnostic data. The MJGHD approach was the only approach that performed well, in terms of classification performance, in both cases. The PGMM approach gave excellent classification results for the Italian wine data but performed poorly when fitted to the breast cancer diagnostic data. Furthermore, the MJGHD approach gave superior classification performance in both cases when compared to the chosen HDDC, MGHFA, and GPCM models.
Although illustrated for clustering, the MJGHD approach can also be applied for semi-supervised classification and discriminant analysis.
In our future work, we will investigate the possibility to avoid calculating the full $\bs \Gamma$ matrix which would greatly reduce the runtime. More efficient code and parallel implementation may also be used for this purpose. 

\section*{Acknowledgements}
This work was partly supported by the Canada Research Chairs program.


\end{document}